\documentclass[runningheads,a4paper]{article}

\usepackage{amsmath}
\usepackage{amsthm} 
\usepackage{amssymb}

\usepackage{txfonts}
\usepackage[utf8]{inputenc}
\usepackage[T1]{fontenc}
\usepackage{lmodern}
\usepackage{hyperref}

\usepackage{graphicx}
\usepackage{color}

\newcommand{\VP}{\textsc{VP}}
\newcommand{\VNP}{\textsc{VNP}}
\newcommand{\Ferm[1]}{\text{Ferm}^{#1}}
\newcommand{\Ham}{\text{Ham}}
\newcommand{\per}{\text{per}}
\newcommand{\CC}{\text{CC}}
\newcommand{\im}{\text{im}}
\newcommand{\p}{\textsc{P}}
\newcommand{\NP}{\textsc{NP}}

\newtheorem{thm}{Theorem}
\newtheorem*{thm*}{Theorem}
\newtheorem{lem}{Lemma}
\newtheorem{prop}{Proposition}
\newtheorem{coro}{Corollary}
\theoremstyle{definition}

\title{Determinant versus Permanent: salvation via generalization?\\
\large{The algebraic complexity of the Fermionant and the Immanant}}
\author{N. de Rugy-Altherre\\
\small{Univ Paris Diderot, Sorbonne Paris Cité,}\\
\small{Institut de Mathématiques de Jussieu, UMR 7586 CNRS,} \\
\small{F-75205 Paris, France}\\
\small{nderugy@math.univ-paris-diderot.fr}} 
\begin{document}
 
\maketitle

\abstract
The fermionant $\Ferm[k]_n(\bar x) = \sum_{\sigma \in S_n} (-k)^{c(\pi)}\prod_{i=1}^n x_{i,j}$ can be seen as a generalization of both the permanent (for $k=-1$) and the determinant (for $k=1$). We demonstrate that it is $\VNP$-complete for any rational $k \neq 1$. Furthermore it is $\#P$-complete for the same values of $k$. The immanant is also a generalization of the permanent (for a Young diagram with a single line) and of the determinant (when the Young diagram is a column). We demonstrate that the immanant of any family of Young diagrams with bounded width and at least $n^{\epsilon}$ boxes at the right of the first column is $\VNP$-complete.

\section{Introduction}
In algebraic complexity (more specifically Valiant's model\cite{Bur00}) one of the main question is to know whether $\VP=\VNP$ or not. Answering this is considered to be a very good step towards the resolution of $P = NP$. This question is very close to the question $\per$ vs. $\det$, where we ask if the permanent can be computed in polynomial time in the size of the matrix, as is the determinant. 

The main idea of this paper is to find a generalization of both the permanent and the determinant in order to study exactly where the difference between them lies. A generalization is here understood as a parameter, let us say $t$, and a function $f(t,\bar x)$ such that for example $f(0,\bar x) = \det(\bar x)$ and $f(1,\bar x) = \per(\bar x)$. If we have a complete classification of the complexity of $f(t,\bar x)$ for any $t$ (with $t$ fixed), we should be able to see where we step from $\VP$ to $\VNP$ and maybe understand a little bit more why the permanent is hard and not the determinant.

Here we study two different generalizations. First the fermionant, secondly the immanant. The fermionant was introduced by Chandrasekharan and Wiese~\cite{CW11} in 2011 in a context of quantum physics. It is defined with a real parameter $k$ such that for $k=1$ it is the determinant and for $k=-1$ it is the permanent. Mertens and Moore~\cite{MM11} have demonstrated its hardness for $k\geq 3$ (and with a weaker hardness for $k = 2$), in the framework of counting complexity.

Likewise, but in a different framework and with a complete different proof, we demonstrate the hardness of the fermionant seen as a polynomial for any rational $k\neq 1$ (and of course for $k \neq 0$). This give a interesting point of view on where the hardness of the permanent lies. We also get a bonus: we use a technique developed by Valiant to demonstrate the hardness of the fermionant in the counting complexity framework for $k \neq 1$. We thus extend the results of Mertens and Moore~\cite{MM11}, in particular to the case $k=2$, which is, from what I understand, the most interesting case for physicists.

The second generalization is more classical and comes from the field of group representation. It is the immanant, introduced by Littlewood~\cite{Lit40} in 1940. Immanants are families of polynomials indexed by Young diagrams. If the Young diagrams are a single column with $n$ boxes, the immanant is the determinant. At the opposite end, if it is a single line of $n$ boxes, the immanant is the permanent. The main question is: for which Young diagrams do we step from $\VP$ to $\VNP$?

We know that if there are only a finite number of boxes on the right of the first column, the immanant is still in $\VP$ (cf \cite{Bur00}). On the other hand, a few hardness results have been found, fundamentally for Young diagrams in which the permanent is hidden. For example, the hook (a line of $n$ boxes and a column of any number of boxes) and the rectangle (any number of lines each with $n$ boxes) are hard (cf~\cite{Bur00}), or more generally if the maximal difference between the size of two consecutive lines is as big as a power of $n$ (cf~\cite{BB03}).  

Here we will demonstrate that for Young diagrams with only two columns, each with $n$ boxes, the immanant is hard, which was an open question (cf~\cite{Bur00} Problem 7.1). As each line of these Young diagrams has length no more than two, the permanent is not hidden in there. More generally for any family of Young diagrams with a bounded number of columns and with at least $n^{\epsilon}$ boxes at the right of the first column, the immanant is hard. It has been conjecture that it is still hard if we remove the bounded condition(cf~\cite{MM11} for example).

For a  complete classification of the immanant in algebraic complexity, one "just" has to determine the complexity of the ziggurat: the Young diagrams where the first line has $n$ boxes, the second $n-1$, the third $n-2$ etc. and the last $1$ box. This immanant is most probably also hard. The complexity of the immanant with a logarithmic number of boxes to the right of the first column is also unknown.

\section{Definitions}

We work within Valiant's algebraic framework. Here is a brief introduction to this complexity theory. For a more complete overview, see~\cite{Bur00}.

An \textit{arithmetic circuit} over $\mathbb Q$ is a labeled directed acyclic connected graph with vertices of indegree $0$ or $2$ and only one sink. The vertices with indegree $0$ are called \textit{input gates} and are labeled with variables or constants from $\mathbb Q$. The vertices with indegree $2$ are called \textit{computation gates} and are labeled with $\times$ or $+$. The sink of the circuit is called the \textit{output gate}.

The polynomial computed by a gate of an arithmetic circuit is defined by induction: an input gate computes its label; a computation gate computes the product or the sum of its children's values. The polynomial computed by an arithmetic circuit is the polynomial computed by the sink of the circuit.

A \textit{p-family} is a sequence $(f_n)$ of polynomials such that the number of variables as well as the degree of $f_n$ is polynomially bounded in $n$. The \textit{complexity} $L(f)$ of a polynomial $f \in \mathbb Q[x_1, \hdots, x_n]$ is the minimal number of computational gates of an arithmetic circuit computing $f$ from variables $x_1, \hdots, x_n$ and constants in $\mathbb Q$. 

Two of the main classes in this theory are: the analog of $\p$, $\VP$, which contains of every p-family $(f_n)$ such that $L(f_n)$ is a function polynomially bounded in $n$; and the analog of $\NP$, $\VNP$. A p-family $(f_n)$ is in $\VNP$ iff there exists a $\VP$ family $(g_n)$ such that for all $n$, 
\[f_n(x_1,\hdots, x_n) = \sum_{\bar{\epsilon} \in \{0,1\}^n} g_n(x_1,\hdots, x_n, \epsilon_1, \hdots, \epsilon_n)\]

As in most complexity theory we have a notion of reduction, the c-reduction: the oracle complexity $L^g(f)$ of a polynomial $f$ with oracle access to $g$ is the minimum number of computation gates and evaluations of $g$ over previously computed values that are sufficient to compute $f$ from the variables $x_1, \hdots x_n$ and constants from $\mathbb Q$. A p-family $(f_n)$ \textit{c-reduces} to $(g_n)$ if there exists a polynomially bounded function $p$ such that $L^{g_{p(n)}}(f_n)$ is a polynomially bounded function. 

$\VNP$ is closed under c-reductions (See~\cite{Poi08} for an idea of the proof). However this reduction does not distinguish lower classes. For example, $0$ is $\VP$-complete for c-reductions. In this paper we will demonstrate hardness results, a smallest notion of reduction (as projection) is thus not needed.

The determinant is in $\VP$. The permanent is $\VNP$-complete for $c$-reductions (\cite{Bur00}).

\section{The fermionant}

Let $A$ be an $n \times n$ matrix. The \textit{fermionant} of $A$, with parameter $k$ is defined as
	\[\Ferm[k] A = \left(-1\right)^n \sum_{\pi \in S_n} \left(-k\right)^{c(\pi)}\prod_{i=1}^n A_{i,\pi(i)}\] 
where $S_n$ denotes the symmetric group of $n$ objects and, for any permutation $\pi \in S_n$, $c(\pi)$ denotes the number of cycles of $\pi$. To study the complexity of such a function, we work within the algebraic complexity framework. The algebraic equivalent of the fermionant is the polynomial obtain where we compute the fermionant on the matrix $(x_{i,j})_{1\leq i,j\leq n}$. If we write $\Ferm[k]$ the p-family $(\Ferm[k]_n)_{n \in \mathbb N}$, we have a complete classification of the algebraic complexity of those polynomials.

\begin{thm}\label{ThmEvalution}
Let $k$ be a rational. 
\begin{itemize}
	\item $\Ferm[0] = 0$.
	\item $\Ferm[1]$ is in $\VP$
	\item for other values of $k$ $\Ferm[k]$ is $\VNP$-complete for c-reductions.
\end{itemize}
\end{thm}

Similarly to the permanent we can see the fermionant as a computation on a graph $G$ with $n$ vertices and the edge between the vertices $i$ and $j$ is labeled with the variable $x_{i,j}$. A permutation $\pi \in S_n$ can be seen as a cycle cover on this graph. A \begin{it}cycle cover\end{it} of $G$ is a subset of its edges that covers all vertices of $G$ and that form cycles. The weight of a cycle cover $\pi$ is $\omega(\pi) = \prod_{e \in \pi}x_e$ and we write $c(\pi)$ its number of cycles, then

\[\Ferm[k](\bar x) = \sum_{\pi \in CC(G)}(-k)^{c(\pi)}\prod_{e\in \pi}x_e\]
where $CC(G)$ is the set of all cycle covers of $G$. We will use a so call iff-gadget, which is the labeled graph draw above. The idea of this gadget is when placed between two edges $e$ and $e'$ on $G$, any cycle cover containing exactly one of the edges $e$ and $e'$ will not contribute to the fermionant computed on the resulting graph.

\begin{center}
\begin{center}\begin{bf}Fig. 1 Burgisser's iff-gadget\end{bf}\end{center}
\includegraphics{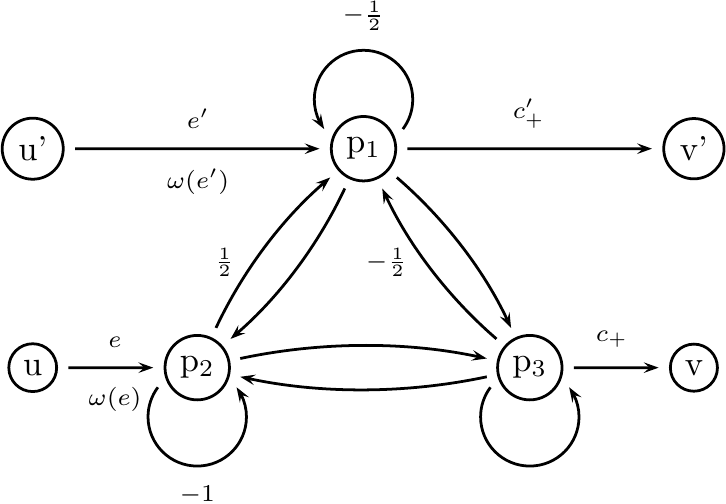} 

\end{center}
\begin{lem}\label{add_several_simultaneously_iff}
Let $G$ be a graph with $n$ vertices and $(e^i_1,e^i_2)_{1\leq i\leq l}$ be a set of pairs of edges of $G$ such that no two edges in this set are equal. Let $G'$ be the same graph but where we place an iff-gadget between every pair $(e^i_1,e^i_2)$. Let $\pi$ be a cycle cover of $G$, $\Pi(\pi)$ be the set of cycle covers of $G'$ that match $\pi$ on $E(G)$.
\begin{itemize}
\item If there is a pair $(e^i_1,e^i_2)$ of edges such that $e^i_1 \in \pi$ and $e^i_2 \notin \pi$, or vice versa, then
$$ \sum_{\pi' \in \Pi(\pi)}(-k)^{c(\pi')}\omega(\pi') =0$$

\item Else, let $d(\pi)$ be the number of pair  $(e^i_1,e^i_2)$ of edges such that $e^i_1\notin \pi$ and $e^i_2 \notin \pi$. Then
$$ \sum_{\pi' \in \Pi(\pi)}(-k)^{c(\pi')}\omega(\pi') = \left(\frac 12(1-k)\right)^{d(\pi)}(-k)^{c(\pi)}\omega(\pi)$$
\end{itemize}
\end{lem}

The proof is in the annexes~\ref{add_several_simultaneously_iff_proof}. Now here is the main tool of our demonstration, that allows us to interpolate the fermionant and compute the permanent.

\begin{lem}\label{copies_and_iff_gadget}
Let $G$ be a graph with $n$ vertices. We make $l$ copies of $G$ and name them $G_1, \hdots, G_l$. Let $\tilde F^l$ be the disjoint union of those copies in which we label the edges of $G_1$ with the same weight as those of $G$ and the edges of $G_i$ for $i \geq 2$ with $1$. If $e$ is an edge of $G$, we call $e_i$ the corresponding edge in $G_i$. We name $F^l$ the graph $\tilde F^l$ where for any edge $e \in E(G)$ and any $1 \leq i \leq l$, we have placed an iff-gadget between $e_i$ and $e_{i+1}$. Let $\pi$ be a cycle cover of $G$ and $\Pi(\pi)$ be the set of cycle covers of $F^l$ that match $\pi$ on $E(G_1)$. Then

$$ \sum_{\pi' \in \Pi(\pi)} (-k)^{c(\pi')}\omega(\pi')  = \left(\frac 12 (1-k)\right)^{(|E(G)|-n)(l-1)} \left(-k\right)^{l\times c(\pi)}\omega\left(\pi\right) $$
\end{lem}





\begin{proof}
The idea is, with the help of the iff-gadget, to copy a cycle cover from $G_1$ to every other copies of $G$, without changing the weight of this cycle cover, just multiplying the number of cycles. The demonstration is by induction on $l$.

If $l=2$, then we simultaneously add $|E(G)|$ iff-gadgets, but only one on each edge. By design, a cycle cover $\pi$ on $G_1$ is repeated on $G_2$ (i.e., if $e_1$ is in $\pi$ then $e_2$ is also in $\pi$ as there is a iff-gadget between $e_1$ and $e_2$. see Lemma~\ref{add_several_simultaneously_iff}). The edges of $G_2$ are labeled with $1$ and therefore do not contribute to the weight of the cycle cover. The number of cycles of $\pi' \in \Pi(\pi)$ is  twice the number of cycles of $\pi$. There is $|E(G)|$ iff-gadgets in $F^2$. A cycle cover of $G$ passes through $n$ edges and therefore activates exactly $n$ iff-gadgets. The other iff-gadgets are not activated and thus each contribute $\frac 12 (1-k)$ to the sum.

Suppose the lemma true for $l-1$ copies. Let $F^{l-1}$ be the disjoint union of $l-1$ copies of $G$ with iff-gadgets. We add a new copy $G_l$ of $G$ linked to $F^{l-1}$ with iff-gadgets to obtain $F^l$. Let $\pi$ be a cycle cover of $G$, $\Pi^l(\pi)$ the set of every cycle covers of $F^l$ that match $\pi$ on $E(G_1)$ and $\Pi^{l-1}(\pi)$ the same but on $F^{l-1}$. By induction,

$$ \sum_{\pi' \in \Pi^{l-1}(\pi)} (-k)^{c(\pi')}\omega(\pi')  = \left(\frac12 (1-k)\right)^{(|E(g)|-n)(l-2)} (-k)^{(l-1)\times c(\pi)}\omega(\pi) $$

Let $\hat F^l $ be the disjoint union of $F^{l-1}$ and $G_l$. To obtain $F^l$ from this graph, one has just to add a iff-gadget between every edge $e_{l-1}$ and $e_l$. We can apply then Lemma~\ref{add_several_simultaneously_iff} to this graph. If $\pi''$ is a cycle cover of $\hat F^l$ that match $\pi$ on $G_1$, let $\Lambda(\pi'')$ be the set of cycle covers of $F^{l}$ that match $\pi''$ on $E(\hat F^{l})$.  Then, if we call $d(\pi'')$ the number of pairs $(e_{l-1},e_l)$ that are not in $\pi''$,
 
 \[ \sum_{\lambda \in \Lambda(\pi'')} (-k)^{c(\lambda)}\omega(\lambda) = \left(\frac 12(1-k)\right)^{d(\pi'')} (-k)^{c(\pi'')}\omega(\pi'')\]
 
 Let us study a little bit more $\pi''$. It is a cycle cover of two disjoint graphs, $F^{l-1}$ and $G_l$. Therefore it is composed of two sub cycle covers: $\sigma'$ a cycle cover of $F^{l-1}$ which by induction is in a $\Pi^{l-1}(\pi)$ and a cycle cover $\lambda$ of $G_l$. However, as every edge of $G_l$ is linked with an iff-gadget to its image in $G_{l-1}$ in $F^l$, the cycle cover $\pi''$ will contribute to the last sum if and only if it contain both $e_{l-1}$ and $e_l$, or neither $e_{l-1}$ and $e_l$. Thus, $\lambda$ must be the copy of $\pi$ in $G_l$, which we write $\lambda_{\pi}$ and $c(\pi'') = c(\sigma)+c(\lambda_{\pi}) = c(\sigma) + c(\pi) $.

There are $n$ edges in the last image $G_l$ that are passed through by $\pi''$. Therefore, there are $(|E(G)|-n)$ iff-gadgets that are not activated by $\pi''$ (i.e., $d(\pi'') = |E(G)|-n$). Thus,
 \begin{align*}
 \sum_{\pi' \in \Pi(\pi)} &(-k)^{c(\pi')}\omega(\pi') = \sum_{\pi'' \in \Pi^l(\pi)}\sum_{\lambda \in \Lambda(\pi'')} (-k)^{c(\lambda)}\omega(\lambda)\\
 &= \sum_{\pi'' \in \Pi^l(\pi)}\left(\frac 12(1-k)\right)^{|E(G)|-n} (-k)^{c(\pi'')}\omega(\pi'')\\
  &=\left(\frac 12(1-k)\right)^{|E(G)|-n}  (-k)^{c(\lambda_{\pi})}\sum_{\sigma \in \Pi^{l-1}(\pi)}(-k)^{c(\sigma)}\omega(\sigma)\\
  &=\left(\frac 12(1-k)\right)^{(|E(G)-n)\times (l-1)} (-k)^{l\times c(\pi)}\omega(\pi)\\
 \end{align*}
 
 Where $\Pi(\pi)$ is the set of cycle covers of $F^l$ that match $\pi$ on $E(G)$; $\Pi^l(\pi)$ the set of cycle covers of $\tilde F^l$ that match $\pi$ on $E(G)$ and for $\pi'' \in \Pi^l(\pi)$, $\Lambda(\pi'')$ the set of cycle covers that match $\pi''$ on $E(\tilde F^l)$. We have $\Pi(\pi) = \bigcup_{\pi'' \in \Pi^l(\pi)}\Lambda(\pi'')$ which completes our demonstration.

\end{proof}

\begin{proof}[Proof of theorem~\ref{ThmEvalution}]
The first case is trivial. For the second, it is a well known result, as $\Ferm[1]_n(\bar x) = \det_n(\bar x)$. Now, let $k$ be a rational different than $ 0$ and $1$. 

Let us write $(P_l G)$ the graph obtained in the previous lemma, when we duplicate $l$ times $G$ and add iff-gadgets to repeat every cycle cover $l$ times. We have seen that

\[\Ferm[k]_{ln}(P_l G)(\bar x) = \sum_{\pi \in \CC(G)} (-k)^{l \times c(\pi)} \prod_{e\in \pi}\omega(e) \times \left(\frac12 \left(1-k\right)\right)^{(l-1)\times(|E(G)| - n)} \]

Let us write $c_m = \sum_{\pi \in \CC(G) | c(\pi) = m} \prod_{e \in \pi} \omega(e)$, $\alpha =  \left( \frac12\left(1-k\right)\right)^{|E(G)|-n}$, $f_l = \Ferm[k]_{ln}(P_{l}(G))$ and $\omega_l = (-k)^{l}$, then
\[
\begin{pmatrix}
   	f_1 \\
	f_2 \\
   	\vdots \\
	f_{n} \\ 
\end{pmatrix} = 
\begin{pmatrix}
\alpha &0 & \hdots & 0 \\
0 & \alpha^2 & \hdots & 0 \\
\vdots & \vdots & \ddots & \vdots \\
0 & 0 & \hdots & \alpha^n\\
\end{pmatrix}
\begin{pmatrix}
   	 \omega_1 & \omega_1^2 & \hdots &  \omega_1^n \\
	\omega_2 & \omega_2^2 & \hdots &\omega_2^n \\
	\vdots & \ddots && \vdots \\
	 \omega_n & \omega_n^2 & \hdots &\omega_n^n \\
\end{pmatrix}
\begin{pmatrix}
   	c_1 \\
	c_2 \\
   	\vdots \\
	c_{n} \\ 
\end{pmatrix}
\]
This system of equation is a Vandermonde system and therefore is invertible (if $k\neq 1$ and $k\neq -1$, because in these cases, some $\omega_i$ are equal and the matrix is not invertible): there exists some rationals $\omega_{l,m}^*$ such that for any $m$, $c_m = \sum_{l=1}^n \omega_{l,m}^* f_l(\bar x) $.

Therefore, for any $m$, we have a c-reduction from $c_m$ to the fermionant, $(c_m) \leq_c (\Ferm[k])$. But, $c_1 := \sum_{\pi \in S_n | c(\pi) = 1} \prod_{i+1}^n x_{i,\pi(i)} = \Ham_n(\bar x)$, where $\Ham_n$ is the Hamiltonian, which is known to be $\VNP$-complete (\cite{Bur00}, Corollary $3.19$). 
\end{proof}

The fermionant can be expressed as a linear combination of polynomial size of the Hamiltonian. From that we have concluded that the fermionant is $\VNP$-complete. However, the Hamiltonian is also $\#P$-complete, when considered as a counting problem. This gives us a Turing reduction from the Hamiltonian to the fermionant and thus it is also $\#P$-complete, but only when computed on rational matrix; the Turing reductions requires rationals ($\frac 12$, $-\frac 1k$, etc). We can adapt the proof of Valiant for the $\#P$-completeness of the permanent to replace those rationals by some gadgets only using $0$ and $1$. And thus we have the following non trivial corollary. The proof is in the annexes~\ref{counting_fermionant_proof}.

\begin{coro}\label{counting_fermionant}
For every $k \neq 1$ and $k \neq 0$, $\Ferm[k]$ is $\#P$-complete for matrices over $\{0,1\}$.
\end{coro}

 \section{Immanant with constant length}
Immanants are defined with characters of representations of $S_n$. Such characters can be indexed by Young diagrams of $n$ boxes (i.e. collections of boxes arranged in left-adjusted rows with a decreasing row length). As all the work of representation theory has already be done (Lemma~\ref{relation_between_ferm_and_imm}), I will not define more those characters. We will only work on Young diagrams.

The \begin{it}immanant\end{it} associated with a Young diagram $Y$ (and its associate character $\chi_Y$) is
 \[\im_{\chi}(\bar x) = \sum_{\pi \in S_n}\chi_Y(\pi)\prod_{i=1}^nx_{i,\pi(i)}\]

For example, if the Young diagram is a single row of $n$ boxes, then for any $\sigma \in S_n$, $\chi_Y(\sigma) = 1$ and thus $\im_Y = \per$. At the opposite end, if $Y$ is a single column with $n$ boxes, $\chi_Y(\sigma) = sg(\sigma)$ and $\im_Y = \det$. For more details (and for a nice demonstration of the Murnaghan-Nakayama rule, one of the main parts of our demonstration), see~\cite{Gre92}.
 
A classical theorem states that the irreducible characters of the symmetric group form a basis for the class functions on $S_n$. Class functions are real functions defined on $S_n$ and stabled under conjugation (i.e. $\forall \pi, \sigma \in S_n, f(\pi \sigma \pi^{-1}) = f(\sigma)$). The function $\pi \longmapsto (-k)^{c(\pi)}$ is such a class function and thus is a linear combination of characters. Mertens and Moore~\cite{MM11} have computed those characters, and applied to the immanant we get:

\begin{lem}\label{relation_between_ferm_and_imm}
For any integers $k$ and $n$, if we write $\Lambda_k^n$ for the set of every Young diagram with $n$ boxes and at most $k$ columns, then there exists some constants $d_{Y}^k$ such that for any matrix $A$:
\[\Ferm[k]_n(A) = \sum_{Y \in \Lambda_k^n} d_{Y}^k\im_{Y}(A)\]
\end{lem}

Intuitively this suggests that the family of every immanants of bounded width is $\VNP$-complete. In algebraic complexity this is not that interesting, as this family is very large. But if we prove that with a certain family of immanant we can compute every immanants of width less than a certain $k$, then this family will be $\VNP$-complete. It is exactly what we are going to do for the demonstration of the following proposition.

\begin{prop}\label{Immanant_2_2}
Let $[n,n]$ be the square Young diagram with two columns, each with $n$ rows. Then $(\im_{[n,n]})_{n \in \mathbb n}$ is $\VNP$-complete for c-reductions.
\end{prop}

\begin{center}
\begin{bf}Young diagrams\end{bf}

\includegraphics{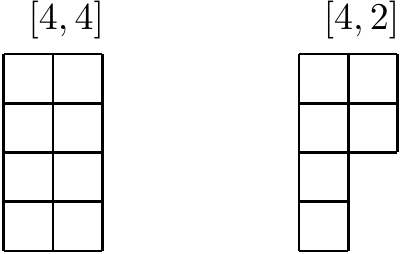} 
\end{center}

\begin{proof}
More generally, let $[l_1,l_2]$ be the two columns Young diagram with $l_1$ boxes in the first column and $l_2$ in the second. More specifically, the Young diagrams of width a most $2$ and of $n$ boxes are $([l,n-l])_{l \in [n/2,n]}$. Each of them can be obtained from the square diagram $[l,l]$ by removing a skew hook of size $\delta = (l - (n-l)) = 2l-n$. A skew hook in a Young diagram is a connected collection of boxes in the border of the diagram such that if you remove this hook it is still a Young diagram (i.e. the row sizes are still decreasing).  Furthermore, if you remove a skew hook of size $\delta$ to $[l,l]$, you can obtain only $[l,n-l]$ and $[l-1,n-l+1]$. The Murnaghan-Nakayama rule (c.f.~\cite{Bur00} chap. 7.2 for more details) states that:

\[\im_{[l,l]}(\bar x, \boldsymbol l) = (-1)^{l-1}\im_{[l,n-l]}(\bar x) + (-1)^l\im_{[l-1,n-l+1]}(\bar x)\]
Where $\boldsymbol l$ is an encoding of a cycle of length $l$. We know that, from Lemma~\ref{relation_between_ferm_and_imm}:

\[\Ferm[2]_n(\bar x) = \sum_{Y \in \Lambda_2^n} d_{Y}^2\im_{Y}(\bar x) = \sum_{l=n/2}^n d_{[l,n-l]}^2 \im_{[l,n-l]}(\bar x)\]

From those two facts, we can compute the fermionant from the square immanant. We just have to take new constants: let $\alpha_{[n-1,1]} = d_{[n-1,1]}(-1)^n$ and for any $ 2 \leq l \leq \frac n2$, $\alpha_{[n-l,l]} = (-1)^l(d_{[l,n-1]} - \alpha_{[l+1,n-l-1]}(-a)^{l+1})$. For simplicity, we write $\alpha_{\boldsymbol l} = \alpha_{[l,n-l]}$. If $n$ is even
\begin{align*}
	\sum_{l = n/2}^{n-1} &\alpha_{\textbf l} \im_{[l,l]}(\bar x, \boldsymbol{2l-n}) \\
	& = \sum_{l = n/2}^{n-1} \alpha_{\textbf l} \left(\left(-1\right)^{l-1}\im_{[l,n-l]}(\bar x) + \left(-1\right)^l\im_{[l-1,n-l+1]}(\bar x) \right)\\ 
	& = d_{[n-1,1]} \im_{[n-1,1]}(\bar x) + \sum_{l=n/2}^{n-2}d_{\textbf l} \im_{[l,n-l]}(\bar x)\\
	& - \sum_{l=n/2}^{n-2} \alpha_{\textbf{l+1}}(-1)^{l+1} \im_{[l,n-l]}(\bar x) + \sum_{l=n/2}^{n-1} \alpha_{\textbf l} (-1)^l\im_{[l-1,n-l+1]}(\bar x)\\
	& = \sum_{l= n/2}^{n-1}d_{\textbf l} \im_{[l,n-l]}(\bar x) - \sum_{l=n/2+1}^{n-1}\alpha_{\textbf l}(-1)^l \im_{[l-1,n-l+1]}(\bar x) \\
	& + \sum_{l=n/2}^{n-1}\alpha_{\textbf l}(-1)^l\im_{[l-1,n-l+1]}(\bar x)\\
	& =  \sum_{l= n/2}^{n-1}d_{\textbf l} \im_{[l,n-l]}(\bar x) + \alpha_{\textbf{n/2+1}} (-1)^{n/2}\im_{[n/2,n/2]}(\bar x)\\ 
\end{align*}

Furthermore, $\im_{[n,0]}(\bar x)=\det_n(\bar x)$ and then can be computed with only a polynomial number of arithmetic operations. 

Thus, $ \sum_{l = n/2}^{n-1} \alpha_{\textbf l} \im_{[l,l]}(\bar x, \boldsymbol{2l-n}) + \det_n(\bar x) + (-1)^{\frac n2}\alpha_{\frac{n}{2}+1} \im_{[\frac n2,\frac n2]}(\bar x)  = \sum_{l = n/2}^n d_{[l,n-l]} \im_{[l,n-l]} = \Ferm[2]_n(\bar x)$.
 
We obtain an arithmetic circuit of polynomial size that compute $\Ferm[2]_n$ with $n/2$ oracles that can compute $\im_{[l,l]}$ for $l \in [n/2,n]$. To obtain a c-reduction from the fermionant to the immant, we just have to notice that $\im_{[l,l]} \leq_p \im_{[l',l']}$ as soon as $l' \geq l$. Indeed, we just have to erase the first $l'-l$-th rows, which can be done by Corollary 3.2 of~\cite{BB03}.
 
The demonstration for $n$ odd works the same, the border cases must just be studied a little bit more closer.

 

\end{proof}

We can generalize this result to almost every family of bounded width. The proof is similar and is in the annex~\ref{thm_immanant_proof}.

\begin{thm}\label{thm_immanant}
Let $(Y_n)$ be a family of Young diagrams of length bounded by $k \geq 2$ such that $|Y_n| = \Omega(n)$. Then
\begin{itemize}
	\item if the number of boxes in the right of the first column if bounded by a constant $c$, then $(\im_{Y_n})$ is in $\VP$.
	\item otherwise, if there is an $\epsilon > 0$ and at least $n^{\epsilon}$ boxes at the right of the first column, $(\im_{Y_n})$ is $\VNP$-complete for c-reductions.
\end{itemize}
\end{thm}

\section{Conclusion and Perspectives}
The generalization via the fermionant tell us that the determinant is really special: the coefficients $1$ and $-1$ allows us, in a simplify way, to cancel some monomials and not to have to compute everything. The $k$ in the fermionant, even thinly different than $1$, separates these monomials and prevents the cancelations.

As for the immanant, the interpretation of the result is harder. Especially as our theorem does not completely classify immanants of constant width, what about the immanant of $[n,\log n]$? Burgisser's algorithm gives a subexponentiel upper bound, but does not put it in $\VP$. Howerver, under the extended Valiant hypothesis (end of chapter 2 in~\cite{Bur00}), it can not be $\VNP$-complete. Is it a good candidate to be neither $\VP$ nor $\VNP$-complete? Or even $\VP$-complete? Or is it as hard as the determinant? This is unknown.

Other generalizations also can be imagined. For example generating functions of a graph property are polynomials that generalize the permanent and some of them can be computed as fast as the determinant. This framework allows us to use our knowledge on graph theory to understand where we step from $\VP$ to $\VNP$. There is no classification of these generation functions, but some results have been found~\cite{Bur00,dRA12}.

I thank to both of my doctoral advisors, A. Durand and G. Malod as well to M. Casula and E. Boulat who try to explain to me the physic behind the fermionant.
\bibliographystyle{alpha}
\bibliography{biblio}

\appendix
\section{Appendix: The iff-gadget}~\label{annex}
We study the iff-gadget. This is done in two steps: first if we add one iff-gadget to the graph (lemma~\ref{add_one_iff}). Then if we add several gadgets but not in the same edges (lemma~\ref{add_several_simultaneously_iff}). The gadget still works if we add several on them, even on a same edge. We will not demonstrate that fact in general, just in the special case that interest us (lemma~\ref{copies_and_iff_gadget}).

\begin{lem}\label{add_one_iff}
Let $G$ be a graph, $\pi$ a cycle cover of this graph. We define its weight as $\omega(\pi) =  \prod_{e\in \pi}x_e$; and we write  $c(\pi)$ for the number of cycles in $\pi$. Let $e$,$e'$ be two edges of $G$. Let $G'$ be the same graph where we place an iff-gadget between these two edges. Let $\Pi(\pi)$ be the set of every cycle cover $\pi'$ of $G'$ which is equal to $\pi$ on $E(G)$. Then 

\begin{itemize}
\item if in $G$, $e$ and $e'$ are in $\pi$ then
$$\sum_{\pi' \in \Pi(\pi)} (-k)^{c(\pi')}\omega(\pi') = (-k)^{c(\pi)}\omega(\pi)$$

\item If $e$ is in $\pi$ and not $e'$, or vis versa,
$$\sum_{\pi' \in \Pi(\pi)} (-k)^{c(\pi')}\omega(\pi') =0$$

\item Finally if neither $e$ nor $e'$ are in $\pi$ then
$$\sum_{\pi' \in \Pi(\pi)} (-k)^{c(\pi')}\omega(\pi')  =  \left(\frac 12 (1-k)\right)(-k)^{c(\pi)}\omega(\pi) $$
\end{itemize}
\end{lem}

\begin{proof}

To demonstrate the completeness of the permanent, Burgisser~\cite{Bur00} has introduced two gadgets: the iff gadget and the Rosette. Our demonstration will be different, however we will use a variant of the iff gadget. Burgisser's iff-gadget is sketched on fig. $1$. If you place this gadget between the edges $e_1$ and $e_2$, it cancel in the calculation of the permanent every cycle cover that pass through $e_1$ but not $e_2$ (or vis versa). This is what we want, however this gadget change the number of cycles of the cycle covers and thus the value of the fermionant. 
\begin{center}
\begin{center}\begin{bf}Fig. 1 Burgisser's iff-gadget\end{bf}\end{center}
\includegraphics{fig1-pics.pdf} 

\end{center}

Fig. 1 is VBurgisser's iff gadget. It adds three vertices $p_1, p_2, p_3$ and connect these vertices according to the following matrix:
\[\begin{pmatrix}
   	-1 & 1 & \frac12 \\
	1 & 1 & -\frac12 \\
	1 & 1 & -\frac12 \\ 
\end{pmatrix} \]
\begin{center}
\begin{center}\begin{bf}Fig.2 Fermionant's iff-gadget\end{bf}\end{center}
\includegraphics{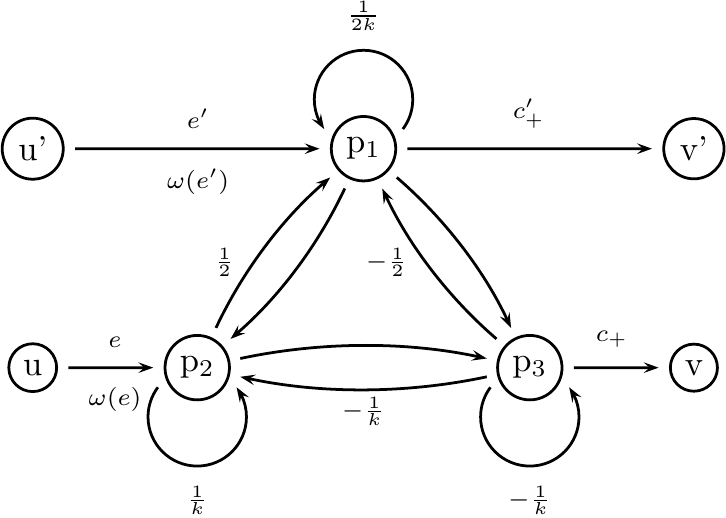} 
\end{center}

Fig. 2 is our iff gadget. We also add three vertices $p_1, p_2, p_3$, but we connect them with the following matrix:
\[\begin{pmatrix}
   	\frac1k & 1 & \frac12 \\
	1 & -\frac1k & -\frac12 \\
	1 & 1 & \frac{1}{2k} \\ 
\end{pmatrix} \]

 Let us now see our iff-gadget. Let $F$ be a graph, $e=(u,v)$ and $e'=(u',v')$ two of its edges and $F'$ the graph built on $F$ by adding our iff-gadget between $e$ and $e'$. For $\pi$ a cycle cover of $F$, let us write $\omega(\pi')$ for the sum of every cycle cover of $F'$ which matches $\pi$ on $F-\{e,e'\}$. Thus,
\begin{itemize}
\item if $e$ and $e'$ are in $\pi$, there is only one cycle covers that cover the gadget with the vertices $c'_-, c'_+, e, c_+$  and then $\omega(\pi')$ = $\omega(\pi)$.
\begin{center}
\begin{center}\begin{bf}Case $e$ and $e'$ in $\pi$\end{bf}\end{center}
\includegraphics{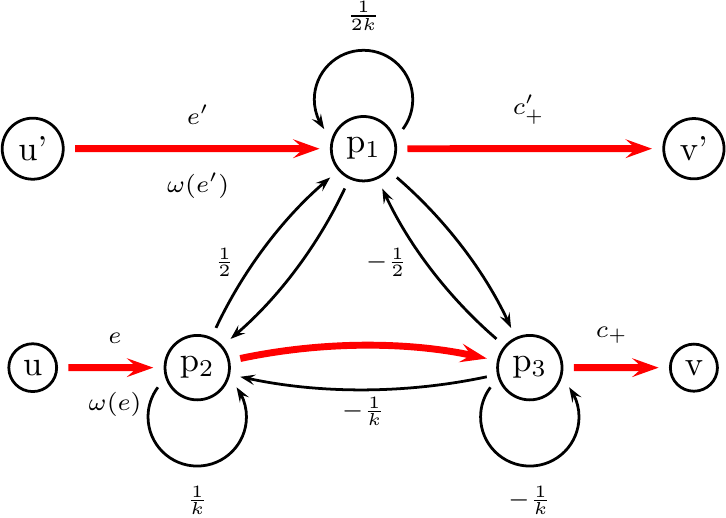} 
\end{center}

\item if $e$ is in $\pi$ but not $e'$, the two possible cycle covers are in blue and red. The red one has as many cycle as $\pi$ and for weight $\frac12\omega(\pi) $. The blue one has one more cycle than $\pi$ and for weight $\frac{1}{2k}\omega(\pi)$. Thus $$(-k)^{c(\pi')}\omega(\pi') = \frac12 (-k)^{c(\pi)}\omega(\pi) + \frac{1}{2k} (-k)^{c(\pi)+1}\omega(\pi) = 0$$
\begin{center}
\begin{center}\begin{bf}Case $e \in \pi$ and $e' \notin \pi$\end{bf}\end{center}
\includegraphics{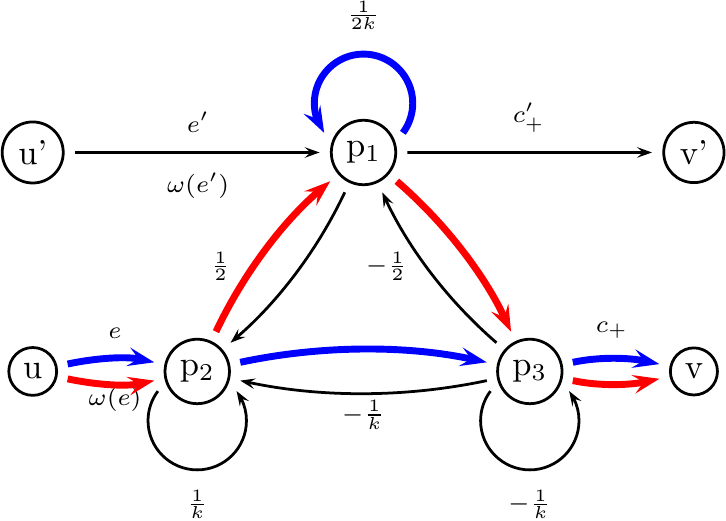} 
\end{center}

\item if $e'$ is in $\pi$ but not $e$, the two possible cycle covers are in blue and red. The red one has one more cycle than $\pi$ and for weight $\frac{1}{-k}\omega(\pi) $. The blue one has two more cycles than $\pi$ and for weight $-\frac{1}{k^2}\omega(\pi)$. Thus $$(-k)^{c(\pi')}\omega(\pi') = \frac{1}{-k} (-k)^{c(\pi)+1}\omega(\pi) - \frac{1}{k^2} (-k)^{c(\pi)+2}\omega(\pi) = 0$$
\begin{center}
\begin{center}\begin{bf}Case $e' \in \pi$ and $e \notin \pi$ \end{bf}\end{center}
\includegraphics{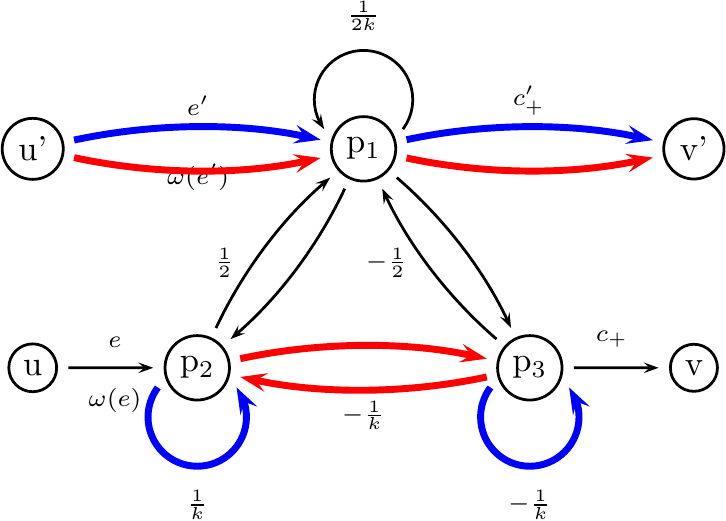} 
\end{center}
\item if neither $e$ nor $e'$ are in $\pi$,  then the six possible cycle covers are listed below. Thus,

\begin{align*}
(-k)^{\pi'}\omega(\pi') &= \left( -\frac{1}{2k^2}(-k)^2  -\frac{1}{2k}(-k)^2   -\frac{1}{2k}(-k)^2   -\frac{1}{2k^3}(-k)^3 - \frac{1}{2}(-k)^1  - \frac{1}{2k}(-k)^1\right) (-k)^{c(\pi)} \omega(\pi)\\
				& =   \left(-\frac{1}{2} -\frac{1}{2}k   -\frac{1}{2}k  +\frac{1}{2} + \frac{1}{2}k  + \frac{1}{2}\right) (-k)^{c(\pi)} \omega(\pi)\\
				& = \frac 12\left(1 - k\right) (-k)^{c(\pi)} \omega(\pi)
\end{align*}

\begin{center}
\begin{center}\begin{bf}Case $e$ and $e'$ not in $\pi$\end{bf}\end{center}
\includegraphics{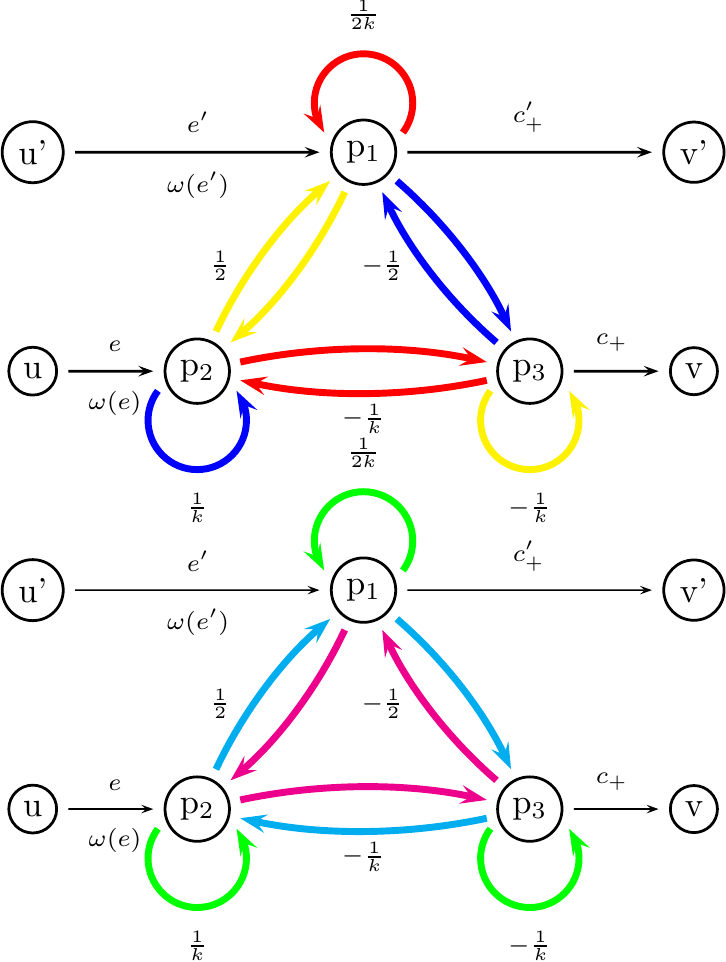} 
\end{center}
\end{itemize}

This end the demonstration. Notices that if $k=-1$ we have Burgisser's iff gadget. If both $e$ and $e'$ are loops (i.e. $u' = v'$ and $u=v$), the gadget is represented in the following scheme and I let the reader check that it works the same way.

\begin{center}
\begin{center}\begin{bf}Loop case\end{bf}\end{center}
\includegraphics{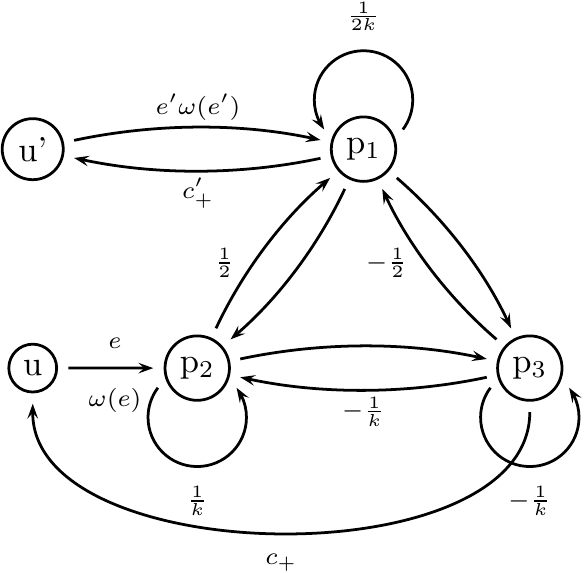} 
\end{center}
\end{proof}

We add now several iff-gadget simultaneously in a graph.
\begin{lem}\label{add_several_simultaneously_iff_proof}
Let $G$ be a graph with $n$ vertices and $(e^i_1,e^i_2)_{1\leq i\leq l}$ be a set of pairs of edges of $G$ such that no two edges in this set are equal. Let $G'$ be the same graph but where we place an iff-gadget between every pair $(e^i_1,e^i_2)$. Let $\pi$ be a cycle cover of $G$, $\Pi(\pi)$ be the set of cycle covers of $G'$ that match $\pi$ on $E(G)$.
\begin{itemize}
\item If there is a pair $(e^i_1,e^i_2)$ of edges such that $e^i_1 \in \pi$ and $e^i_2 \notin \pi$, or vis versa, then
$$ \sum_{\pi' \in \Pi(\pi)}(-k)^{c(\pi')}\omega(\pi') =0$$

\item Else, let $d(\pi)$ be the number of pair  $(e^i_1,e^i_2)$ of edges such that $e^i_1\notin \pi$ and $e^i_2 \notin \pi$. Then
$$ \sum_{\pi' \in \Pi(\pi)}(-k)^{c(\pi')}\omega(\pi') = \left(\frac 12(1-k)\right)^{d(\pi)}(-k)^{c(\pi)}\omega(\pi)$$
\end{itemize}
\end{lem}

\begin{proof}
By induction on $l$. If $l=1$, i.e., we add only one iff-gadget, then this case has been dealt with in Lemma~\ref{add_one_iff}. Suppose the lemma true for $l-1$ iff-gadgets. Let $\pi$ be a cycle cover of $G$. 

\begin{itemize}
\item If we are in the first case, then let $(e^i_1,e^i_2)$ be a pair of edges such that $e^i_1 \in \pi$ and $e^i_2 \notin \pi$. Let $G^i$ be the graph with all the same iff-gadgets than $G'$ but the one between $e^i_1$ and $e^i_2$. By induction the sum of every cycle covers of $G^i$ that match $\pi$ on $E(G)$ is $0$ or $ \left(\frac 12(1-k)\right)^{d(\pi)}(-k)^{c(\pi)}\omega(\pi)$. Let $\pi'$ one of those cycle covers. Remarque that $e^i_1 \in \pi'$ and $e^i_2 \notin \pi'$.

We add an iff-gadget between $e^i_1$ and $e^i_2$ in $G^i$ and obtain $G'$. We can consider $G'$ as a graph with only one iff-gadget and apply lemma~\ref{add_one_iff} with this graph and $\pi'$. Then the sum of every cycle covers that match $\pi'$ on $E(G^i)$ is $0$. Thus the result.

\item If we are in the second case, i.e., every edges of the same pair are simultaneously in or out of $\pi$. Let $G^l$ be the same graph as $G'$ but with no iff-gadget between $e^l_1$ and $e^l_2$. By induction, if we write $\Pi^l(\pi)$ the set of cycle cover of $G^l$ that match $\pi$ on $E(G)$ and $d^l(\pi)$ the number of pair that are in $\pi$, but $(e^l_1,e^l_2)$,
$$ \sum_{\pi' \in \Pi^l(\pi)}(-k)^{c(\pi')}\omega(\pi') = \left(\frac 12(1-k)\right)^{d^l(\pi)}(-k)^{c(\pi)}\omega(\pi)$$

We now see $G'$ as a graph with only one iff-gadget, the one between $e^l_1$ and $e^l_2$. Let $\pi' \in \Pi^l(\pi)$ and $\Lambda$ be the set of cycle covers of $G'$ that match $\pi'$ on $E(G')$ minus every edges of this iff-gadget. If $e_1^l$ and $e_2^l$ are in $\pi$, then they are in $\pi'$ and 
$$ \sum_{\lambda \in \Lambda(\pi')}(-k)^{c(\lambda)}\omega(\lambda) = (-k)^{c(\pi')}\omega(\pi')$$

Thus,
\begin{align*}
\sum_{\pi' \in \Pi(\pi)}(-k)^{c(\pi')}\omega(\pi')&=\sum_{\pi' \in \Pi^l(\pi)}\sum_{\lambda \in \Lambda(\pi')} (-k)^{c(\lambda)}\omega(\lambda)\\
& = \sum_{\pi' \in \Pi^l(\pi)} (-k)^{c(\pi')}\omega(\pi')\\
& = \left(\frac 12(1-k)\right)^{d(\pi)}(-k)^{c(\pi)}\omega(\pi)
\end{align*}

If $e_1^l$ and $e_2^l$ are not in $\pi$, then they are not in $\pi'$ and

$$ \sum_{\lambda \in \Lambda(\pi')}(-k)^{c(\lambda)}\omega(\lambda) = \left(\frac 12(1-k)\right)(-k)^{c(\pi')}\omega(\pi')$$

Thus,
\begin{align*}
\sum_{\pi' \in \Pi(\pi)}(-k)^{c(\pi')}\omega(\pi')&=\sum_{\pi' \in \Pi^l(\pi)}\sum_{\lambda \in \Lambda(\pi')} (-k)^{c(\lambda)}\omega(\lambda)\\
& = \sum_{\pi' \in \Pi^l(\pi)}  \left(\frac 12(1-k)\right)(-k)^{c(\pi')}\omega(\pi')\\
& = \left(\frac 12(1-k)\right)^{d^l(\pi)+1}(-k)^{c(\pi)}\omega(\pi)\\
& = \left(\frac 12(1-k)\right)^{d(\pi)}(-k)^{c(\pi)}\omega(\pi)
\end{align*}

\end{itemize}
\end{proof}

\section{Appendix: Counting Fermionant}\label{counting_fermionant_proof}
The proof here is similar as Valiant's~\cite{Val792}, but with every gadgets adapt to work with the fermionant, i.e., each gadget is design to compensate the number of cycles it add, or to add a constant number of cycles.
\begin{thm}
For every $k \neq 1$ and $k \neq 0$, $\Ferm[k]$ is $\#P$-complete for matrices over $\{0,1\}$
\end{thm}

\begin{proof}
Let $k $ be a integer neither null nor equal to $1$. We have demonstrated in Theorem~\ref{ThmEvalution} that there exists rationals $w_i^*$ such that for every rational matrix $A$:

$$ \Ham(A) = \sum_{i=1}^n w_i^* \Ferm[k](P_iA) $$

$P_l$ is a transformation on matrices which use rational constants (it is the transformation that duplicate the underlying graph $i$-th time and add iff-gadgets). Then, even if the Hamiltonian is known to be $\#P$-complete for $\{0,1\}$ matrices, we can only conclude from this equality that the fermionant is $\#P$-complete for rational matrices. We must work harder in order to obtain the $\#P$-completeness for $\{0,1\}$-matrices, first by obtaining the completeness for nonnegative integers, then any numbers and at last for $\{0,1\}$. 

The only  non-integer rationals used in the computation of $\Ferm[k](P_iA)$ are $\frac{1}{k}$, $\frac12$, $\frac{1}{2k}$ and their opposites, if there are no non-integer rationals in $A$. To get rid of those, we just have to multiply every edge weight of $P_iA$ by $2k$. Then, if we call $n_i$ the size (i.e., the number of rows) of $P_i A$ and $P_i' A$ the new matrix, which is an integer matrix, we have:

$$\Ham(A) =  \sum_{i=1}^n w_i^* (\frac{1}{2k})^{n_i} \Ferm[k](P_i' A) $$

From this new equation, we can say that the fermionant is $\#P$-complete for integer matrices. For the rest of the demonstration, we will suppose that $A$ is a $\{0,1\}$-matrix. To go from integers to natural numbers, i.e., to get ride of the $-1$, we will use a classical method, used by Valiant~\cite{Val792} to demonstrate the $\#P$-completeness of the permanent: we compute everything modulo a great number $\Lambda$. Therefore, we can switch every $-1$ into $\Lambda+1$.

 The biggest integer in $P_l' A$ is $2k$. Then, if we call $n_i$ the size of $P_i(A)$,
 
 \begin{align*}
 \Ferm[k](P_l'A) & = \sum_{\pi \in S_{n_i}} (-k)^{c(\pi)} \prod_{j=1}^{n_i} (P_i'A)_{j, \pi(j)}\\
 			& \leq \sum_{\pi \in S_{n_i}} k^{n_i} \prod_{j=1}^{n_i} 2k\\
			& \leq n_i! k^{n_i} (2k)^{n_i}
 \end{align*}
 
To make sur that $\Lambda$ is bug enough, we take $\Lambda >  \sum_{i=1}^n |w_i^*|n_i!(2k)^{2n_l}$. Let us switch in every $P_i' A$ every negative integer $-m$ by $(\Lambda+1)\times m$ and let us call $P_i''A$ the result. It is a matrix of natural numbers. Therefore

$$ \Ham(A) \equiv  \sum_{i=1}^n  w_i^* (\frac{1}{2k})^{n_i}  \Ferm[k](P_i'' A) \left[ \Lambda \right]$$

This equation has of course no meaning, as $w_i^* (\frac{1}{2k})^{n_i} $ is a rational and cannot be computed modulo $\Lambda$. Let  $w_i^* = \frac{p_i}{q_i}$, where $p_i$ and $q_i$ are integers and $\omega = \prod_{i=1}^n q_i^* $. Thus $\check w_i^* := \omega \times w_i^*$ is an integer. Let also $m_i = n_n - n_i$, where $n_i$ is the size of $P'_i A$. The right equation is then

$$ \Ham(A) \times (2k)^{n_n} \times \omega \equiv  \sum_{i=1}^n \check w_i^* (2k)^{m_i}  \Ferm[k](P_i'' A) \left[ \Lambda \right]$$

The next step is to replace every edge labeled by a natural number $a$ by a gadget with only edges labeled by $1$. In Fig. $4$ we give an example of this gadget for an edge $e$ between $u$ and $v$ labeled with $a = 20 = 2^2+2^4$. A cycle cover $\pi$ that pass thought $e$ in the graph now can go in $16$ different ways thought $p_1^4, p_2^4$ and $p_3^4$ and in $4$ different ways trought $p_1^2$. We want that the number of cycles this gadget add to be the same, whatever path $\pi$ took. The loop-gadgets (centered in $q_i^j$) are here for that.

\begin{center}
\begin{center}\begin{bf}Fermionant's gadget\end{bf}\end{center}
\includegraphics{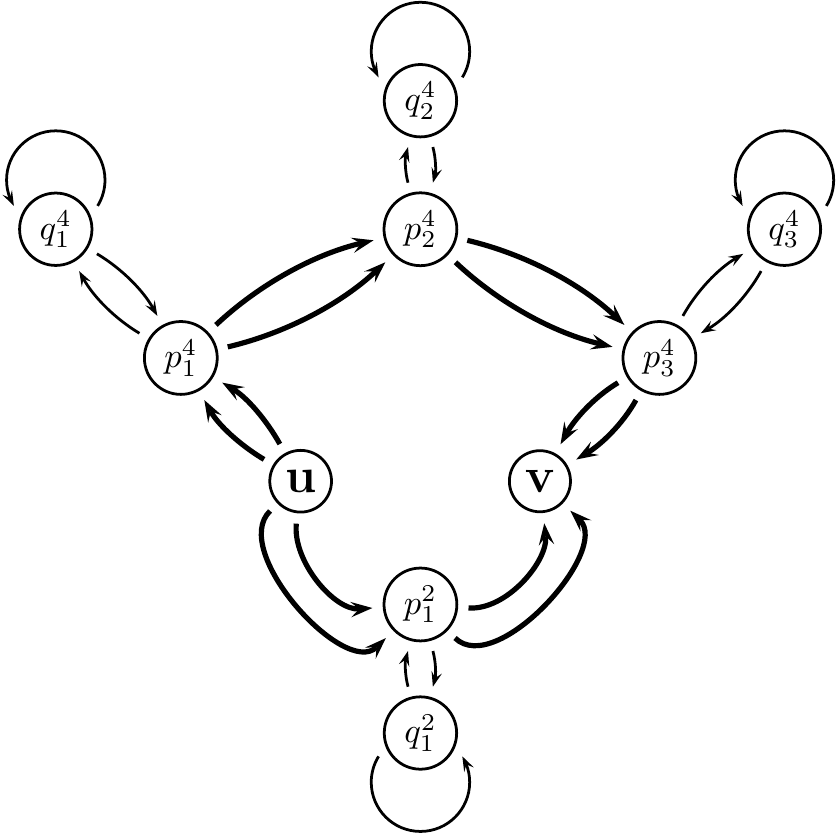} 
\end{center}
Now instead of an edge that counts in a cycle cover for a weight $a$ we have $a$ different cycle covers that pass through this edge. However, this gadget needs $\{0,1,2\}$ matrices to represent the double edges. We have to replace in this gadget the double edges by a diamond-gadget:
\begin{center}
\begin{center}\begin{bf}Diamond gadget\end{bf}\end{center}
\includegraphics{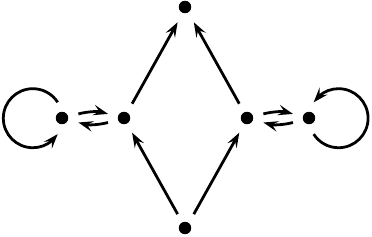} 
\end{center}

The idea of this diamond gadget is obvious. Here also the loop-gadgets are place to keep constant the number of cycles we add. 

Let $e$ be an edge between two vertices $u$ and $v$ labeled by a non negative integer $a$ in a graph $G$. In base $2$, $a$ is written 
$$a = \sum_{i=1}^{|a|}a_i2^i$$

We delete the edge $e$. For every $i \leq |a|$ such that $a_i \neq 0$, we add $i-1$ new vertices $p_1^i, \hdots, p_{i-1}^i$ and we connect $p_j^i$ to $p_{j+1}^i$ with a diamond gadget. The first of this new vertices is connected similarly with $u$ and the last with $v$. Finally we add on every new vertex $p_j^i$ a new looped vertex $q_j^i$ linked by a cycle to $p_j^i$ (i.e., we add on every new vertix a loop-gadget).

\begin{center}
\begin{center}\begin{bf}Loop gadget\end{bf}\end{center}
\includegraphics{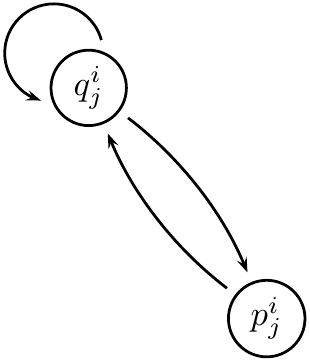}
\end{center}

Let $G'$ the graph where $e$ has been replaced by the gadget $M(a)$. If $\pi$ is a cycle cover of $G$, let us write $\Check \omega(\pi)$ the sum of every cycle cover $\pi'$ in $G'$ that match $\pi$ in $G'-\{e\}$.  

If $e \in \pi$ then $\omega(\pi)  (-k)^{|a|+2(|a|+1)} = \Check \omega(\pi)$, where $|a|$ is the number of bits of $a$. Indeed, the weight is the same but we add with our gadget $|a|+1$ diamond gadgets and $|a|$ loop gadgets. Therefore, we add $2(|a|+1) + |a|$ new cycles. Furthermore, because of the loop-gadget we add, even if $e \notin \pi$ we still have  $\omega(\pi)  (-k)^{|a|+2(|a|+1)}  = \Check \omega(\pi)$. If both cases the final result is multiplied by $(-k)^{|a|+2(|a|+1)} $. 

If we switch every labeled edge in $P_i'' A$ by a similar gadget, we multiply the result by a power of $(-k)$. However, this power does not depend on $A$, as it is a $\{0,1\}$-matrix.

Indeed let us see where the edges labeled by something else that $0$ or $1$ are in $P_l''A$.. Let us call $G$ the complete graph of size $n$ (i.e. the graph where we put the labels given by $A$). We have performed the following transformations on $G$:
\begin{itemize}
	\item $G_l$: we duplicate $G$ $l$ times.
	\item $G_l'$: we add to every vertex a loop-gadget.
	\item $G_l''$: we add on every edge which is not in a loop-gadget an iff-gadget.
	\item $G_l'''$: we add on every main edge of a loop-gadget an iff-gadget.
	\item $G_l^{\textit{IV}}$: we multiply every label by $2k$.
	\item $G_l^{\textit{V}}$: we change every negative label by its opposite multiplied by $\Lambda$.
\end{itemize}

Let us summarize for every gadget the non zero and non $1$ edges we have. Let us noticed that there are no edges beside those on gadgets.

\begin{itemize}
	\item For an iff-gadget put on an edge $e$ which is not on a loop-gadget. $e$ therefore has a weight $0$ or $1$ in $G$. The iff-gadget has:
	\begin{itemize}
		\item $8$ edges labeled by $2k$
		\item $1$ edge labeled by $k$
		\item $1$ edge labeled by $k\Lambda$
		\item $1$ edge labeled by $2$
		\item $1$ edge labeled by $2\Lambda$
		\item $1$ edge labeled by $1$ 
	\end{itemize}
	\item For an iff-gadget put on an edge $e$ which is on a loop-gadget. $e$ therefore has a weight $-\frac1k$ in $G$. The iff-gadget has:
	\begin{itemize}
		\item $6$ edges labeled by $2k$
		\item $1$ edge labeled by $k$
		\item $1$ edge labeled by $k\Lambda$
		\item $1$ edge labeled by $2$
		\item $3$ edges labeled by $2\Lambda$
		\item $1$ edge labeled by $1$ 
	\end{itemize}
	\item For a loop-gadget. The main edge of this gadget is accounted for in the previous iff-gadget. The loop-gadget has:
	\begin{itemize}
		\item $1$ edge labeled by $2k$
		\item $1$ edge labeled by $2\Lambda$
	\end{itemize}
\end{itemize}

There are $l\times n$ loop-gadgets, $(l-1) \times n$ iff-gadgets put on a loop-gadget and $\frac{n(n-1)}2 \times (l-1)$ iff gadgets put a on regular edge in $G_l^{\textit{V}}$. We therefore multiply the final result by $(-k)^{\gamma_l}$ where 
$\gamma_l = \left(8|2k| + |k| + |K\Lambda| + |2| + |2\Lambda|\right) \left(\frac{n\left(n-1\right)}2  \left(l-1\right)\right) + \left(6|2k| + |k| + |K\Lambda| + |2| +3 |2\Lambda|\right)\left(l-1\right)n + \left(|2k| + |2\Lambda|\right)$.

Thus if we write $P'''_l A$ the $\{0,1\}-$matrix obtained from $P_i'' A$ by replacing every edge labeled with $a\neq0,1$ by the gadget $M(a)$, 

$$\Ham(A) \times (2k)^{n_n} \times \frac1{\omega} \equiv  \sum_{i=1}^n \check w_i^* (2k)^{m_i} (\frac{1}{-k})^{\gamma_l}  \Ferm[k](P'''_i A) \left[ \Lambda\right]$$

As before, this equation has no meaning, as we try to compute rationals modulo $\Lambda$. But, if we write $\gamma = \max \gamma_i$,

$$\Ham(A) \times (2k)^{n_n} \times (-k)^{\gamma}` \times \frac1{\omega} \equiv  \sum_{i=1}^n \check w_i^* (2k)^{m_i} (-k)^{\gamma - \gamma_l}  \Ferm[k](P'''_i A) \left[ \Lambda\right]$$

And therefore the $\Ferm[k]$ is $\#P$-complete for $\{0,1\}$ matrices. Note that $P_i''' A$ has a size polynomialy bounded in the size of $A$.

\end{proof}

\section{Appendix: Immanant, proof of Theorem~\ref{thm_immanant}}\label{thm_immanant_proof}
We have demonstrate in Lemma~\ref{Immanant_2_2} that the family of immanants of a $[n,n]$ Young diagrams is hard. Here is a generalization needed to demonstrate Theorem~\ref{thm_immanant}. 

\begin{prop}
Let $(Y_n)$ be a family of Young diagrams with two columns and at least $\Omega(n^{\epsilon})$ boxes in the last one, for $\epsilon > 0$. Then $(\im_{Y_n})$ is $\VNP$-complete for c-reductions.
\end{prop}
\begin{proof}
For $m \in \mathbb N$ let $k_1^m$ be the size of the first column of $Y_m$ and $k_2^m$ the size of the second column. We will in this proof consider the case where the difference between those two sizes grows like $n$ and the case where it is constant.

\begin{itemize}
\item If $(k_1^m-k_2^m) > k_2^m$ (i.e. $k_1^m > 2 k_2^m$), then there is only one skew hook of size $\delta := k_1^m-k_2^m$ that can be removed from $[k_1^m,k_2^m]$ (in red in the picture). If we apply the Murnaghan-Nakayama Rule,

\[\im_{[k_1^m,k_2^m]}(\bar x, \sigma_{\delta}) = (-1)^{\delta-1}\im_{[k_2^m,k_2^m]}(\bar x)\]

Where $\sigma_{\delta} $ is a cycle of size $\delta$. Thus, $(\im_{[k_2^m,k_2^m]}) \leq_c (\im_{[k_1^m,k_2^m]})$ and it is $\VNP$-complete by proposition~\ref{Immanant_2_2}

\begin{center}
\includegraphics{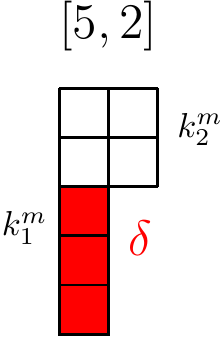} 
\end{center}

\item If $k_2^m \leq k_1^m \leq 2k_2^m$ and $\delta^m := (k_1^m-k_2^m) = \Omega(n^{\epsilon)}$ for an $\epsilon > 0$. First, by using Corollary 3.2 of \cite{BB03} (the projection they referring to is a type of c-reduction, we remove the first $k_2^m - \delta^m$ lines from $[k_1^m,k_2^m]$ (as $\delta^m \leq k_2^m$). We obtain the Young diagram $[2\delta^m,\delta^m]$. Then we apply the Murnaghan-Nakayama Rule, with a skew hook of size $\delta^m$. 

\[\im_{2\delta^m,\delta^m}(\bar x, \sigma_{\delta^m}) = (-1)^{\delta^m-1}\left(\im_{[\delta^m,\delta^m]}(\bar x) + \im_{[2\delta^m,0]}(\bar x)\right)\]

But, $\im_{[2\delta^m,0]}(\bar x) = \det(\bar x)$ and it can be computed with a number of operation polynomial in $\delta$. Therefore, with only a polynomial number of additions, multiplications and evaluations of $\im_{[k_1^m,k_2^m]}$ we can compute $\im_{[\delta^m,\delta^m]}$, which is $\VNP$-complete by the proposition~\ref{Immanant_2_2}.

\item Else, i.e., if $k_2^m \leq k_1^m \leq 2k_2^m$ and  $\delta^m := (k_1^m-k_2^m)$ is bounded by a constant.  We fixe $m \in \mathbb N$ and we write $\delta$ for $\delta^m$, $k_1$ for $k_1^m$ and $k_1$ for $k_2^m$.

We apply the Murnaghan-Nakayama Rule with a skew hook of size $2l_n+\delta +2$. There are only two such skew hooks:
\begin{align*}
&\sum_{l=\frac n2 + \delta+1}^{n - \delta -1}a_{[l,n-l]}\im_{[l+1+\delta,n-l]}(\bar x, \sigma_{2l-n+\delta+2}) \\
= &\sum_{l=\frac n2 + \delta+1}^{n - \delta -1} a_{[l,n-l]}(-1)^{2l-n+\delta+1} \left(\im_{[l,n-l]}+ \im_{[l+\delta+1,n-l-\delta-1]}\right) \\
\end{align*}

We put together the same immanants.
\begin{align*}
= &(-1)^{\delta-n+1}\left(\sum_{l=\frac n2 + \delta+1}^{n - \delta -1} a_{[l,n-l]}\im_{[l,n-l]} + \sum_{p=\frac n2 + 2\delta+2}^{n} a_{[p-\delta-1,n-p+\delta+1]}\im_{[p,n-p]}\right) \\
= &(-1)^{\delta-n+1}\left(\sum_{l=\frac n2 + 2\delta+2}^{n - \delta -1}\im_{[l,n-l]} \left(a_{[l,n-l]} + a_{[l-\delta-1,n-l+\delta+1]}\right) + \sum_{l=\frac n2+\delta+1}^{\frac n2 +2\delta+1}a_{[l,n-l]}\im_{[l,n-l]}\right.\\
&\left.+ \sum_{p=n-\delta}^n a_{[p-\delta-1,n-p+\delta+1]}\im_{[p,n-p]}\right) \\
\end{align*}

The $a_{[i,j]}$ have been built such that $ \left(a_{[l,n-l]} + a_{[l-\delta-1,n-l+\delta+1]}\right) = d_{[l,n-l]}$. Ergo,
\begin{align*}
=  (-1)^{\delta-n+1}\left(\sum_{l=\frac n2 + \delta+1}^{n - \delta -1} \im_{[l,n-l]} d_{[l,n-l]} \right.&\left.+ \sum_{l=0}^{\delta}a_{[\frac n2 +\delta+1+l,\frac n2 -\delta - 1 - l]}\im_{[\frac n2 +\delta+1+l,\frac n2 -\delta - 1 - l]}\right.\\
&\left. + \sum_{l=0}^{\delta}a_{[n-2\delta-1+l,2\delta+1-l]}\im_{[n-2\delta-1+l,2\delta+1-l]}\right)\\
\end{align*}

And by Lemma~\ref{relation_between_ferm_and_imm}, we have

\begin{align*}
=   (-1)^{\delta-n+1}\left(\Ferm[2]_n -\right.&\left. \sum_{l=\frac n2}^{\frac n2 + \delta}d_{[l,n-l]} \im_{[l,n-l]} - \sum_{l=n-\delta}^nd_{[l,n-l]} \im_{[l,n-l]} \right)\\
& +  (-1)^{\delta-n+1}\left( \sum_{l=0}^{\delta}a_{[\frac n2 +\delta+1+l,\frac n2 -\delta - 1 - l]}\im_{[\frac n2 +\delta+1+l,\frac n2 -\delta - 1 - l]}\right)\\
& +  (-1)^{\delta-n+1}\left( \sum_{l=0}^{\delta}a_{[n-2\delta-1+l,2\delta+1-l]}\im_{[n-2\delta-1+l,2\delta+1-l]}\right)\\
=   (-1)^{\delta-n+1}\left(\Ferm[2]_n -\right.&\left. \sum_{l=0}^{\delta}d_{[\frac n2+l,\frac n2-l]} \im_{[\frac n2+l,\frac n2-l]} - \sum_{l=0}^{\delta}d_{[n-\delta+l,\delta-l]} \im_{[n-\delta+l,\delta-l]} \right)\\
& +  (-1)^{\delta-n+1}\left( \sum_{l=0}^{\delta}a_{[\frac n2 +\delta+1+l,\frac n2 -\delta - 1 - l]}\im_{[\frac n2 +\delta+1+l,\frac n2 -\delta - 1 - l]}\right)\\
& +  (-1)^{\delta-n+1}\left( \sum_{l=0}^{\delta}a_{[n-2\delta-1+l,2\delta+1-l]}\im_{[n-2\delta-1+l,2\delta+1-l]}\right)
\end{align*}

We have to compute from $\im_{[k_1,k_2]}$ the residual terms. $ \sum_{l=0}^{\delta}a_{[n-2\delta-1+l,2\delta+1-l]}\im_{[n-2\delta-1+l,2\delta+1-l]}$ and $\sum_{l=0}^{\delta}d_{[n-\delta+l,\delta-l]} \im_{[n-\delta+l,\delta-l]}$ have both less than $2\delta+2$ boxes in the last column. With the algorithm found in~\cite{Bur00}, we can compute those immanants in a number of operations polynomially bounded in $n$ but exponential in $\delta$. However, $\delta$ is here a constant. Therefore, those residual terms can be counted in a polynomial number of operations.

Now, let us look at 
\[\sum_{l=0}^{\delta}a_{[\frac n2 +\delta+1+l,\frac n2-\delta-1-l]}\im_{[\frac n2 +\delta+1+l,\frac n2 -\delta - 1 - l]}\]
We use a technique similar to before.
\begin{align*}
&\sum_{l=0}^{\delta}a_{[\frac n2+\delta+1+l,\frac n2-\delta-1-l]}\im_{[\frac n2+\delta+1+l,\frac n2+1+l]}(\bar x, \sigma_{\delta+2l+2})\\
&=\sum_{l=0}^{\delta}a_{[\frac n2+\delta+1+l,\frac n2-\delta-1-l]}\left((-1)^{\delta+2l+1}\im_{[\frac n2+\delta+1+l,\frac n2-1-\delta-l]}+(-1)^{\delta+2l+2}\im_{[\frac n2+l,\frac n2-l]}\right)
\end{align*}

Indeed, there are two skew hooks of size $\delta+2l+2$ that we can remove from $[\frac n2+\delta+1+l,\frac n2+1+l]$: the one in the second column and the one represented in red in the graphic.
\begin{center}
\includegraphics{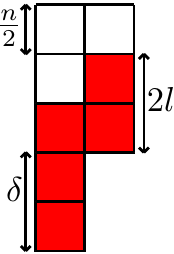} 
\end{center}

Therefore, from $\im_{[\frac n2+2\delta+1,\frac n2+1+\delta]}$ we can compute in only a polynomially number of operations 
\[\sum_{l=0}^{\delta}a_{[\frac n2+\delta+1+l,\frac n2-\delta-1-l]}\im_{[\frac n2+\delta+1+l,\frac n2-1-\delta-l]}+\sum_{l=0}^{\delta}a_{[\frac n2+\delta+1+l,\frac n2-\delta-1-l]}\im_{[\frac n2+l,\frac n2-l]}\]
The first term is the one we were trying to compute, the second one can be added to the last residual term. Finally, we have to compute from $\im_{[k_1,k_2]}$ 

\[ \sum_{l=0}^{\delta}(d_{[\frac n2+l,\frac n2-l]} + a_{[\frac n2+\delta+1+l,\frac n2-\delta-1-l]}) \im_{[\frac n2+l,\frac n2-l]}\]

We use the same technique.
\begin{align*}
\sum_{l=0}^{\delta}b_{[\frac n2+l, \frac n2-l]} &\im_{[\frac n2-l+\delta, \frac n2-l]}(\bar x, \sigma_{\delta -2l}) \\
=& \sum_{l=0}^{\delta}b_{[\frac n2+l, \frac n2-l]} (-1)^{\delta-2l+1}\left(\im_{[\frac n2+l,\frac n2-l]} + \im_{[\frac n2-l+\delta,\frac n2 -\delta + l]}\right)\\
=&  (-1)^{\delta+1}\left(\sum_{l=0}^{\delta}b_{[\frac n2+l, \frac n2-l]}im_{[\frac n2+l,\frac n2-l]} + \sum_{p=0}^{\delta}b_{[\frac n2 + \delta - p, \frac n2+p-\delta]}\im_{[\frac n2+p,\frac n2-p]} \right)\\
=&  (-1)^{\delta+1}\sum_{l=0}^{\delta}im_{[\frac n2+l,\frac n2-l]}\left(b_{[\frac n2+l,\frac n2-l]}+b_{[\frac n2 + \delta - p, \frac n2+p-\delta]} \right)\\
=& (-1)^{\delta+1} \sum_{l=0}^{\delta}(d_{[\frac n2+l,\frac n2-l]} + a_{[\frac n2+\delta+1+l,\frac n2-\delta-1-l]}) \im_{[\frac n2+l,\frac n2-l]}\\
\end{align*}

\end{itemize}
\end{proof}

\begin{thm}
Let $(Y_n)$ be a family of Young diagram of length bounded by $k \leq 2$ such that: $|Y_n| = \Omega(n)$. Then
\begin{itemize}
	\item if the number of boxes in the right of the first column if bounded by a constant $c$, then $(\im_{Y_n})$ is in $\VP$.
	\item otherwise, if there is at least $n^{\epsilon}$ boxes at the right of the first column, $(\im_{Y_n})$ is $\VNP$-complete for c-reductions.
\end{itemize}
\end{thm}

\begin{proof}
The first case is a result of the algorithm found in~\cite{Bur00}. For the second case, we can suppose that every column counted, i.e., that the number of boxes in the last row is of $\Omega(n^{\epsilon})$ for an $\epsilon > 0$. Indeed, if it is not, let $l < k$ be such that the number of boxes in the last $l$ columns is bounded by a constant $c$ in every $Y_n$ but the number of boxes in the $l+1$ last column is not. Then if we remove from every $Y_n$ the first $c$ rows then for every $n$, $Y_n$ has no boxes in the last $l$ columns and $\Omega(n^{\epsilon})$ in the $l+1$ last column.

Now that every column has at a nonconstant number of boxes, especially the last two, let us remove the first $k-2$ columns. The Young diagrams we obtain, $\mu_n$, have only two columns but at least $\Omega(n^{\epsilon})$ boxes for an $\epsilon >0$. And the second column has a non-bounded number of boxes. Therefore, we are in the case of $k=2$ and $(\im_{\mu_n})$ is $\VNP$-complete. Furthermore, we have, by Corollary 3.2 of~\cite{BB03}, $(\im_{\mu_n}) \leq_p (\im_{Y_n})$. Thus $(\im_{Y_n})$ is $\VNP$-complete.

\end{proof}

\end{document}